\documentclass[aps,prx,twocolumn,superscriptaddress]{revtex4-2}

\usepackage[T1]{fontenc}
\usepackage[utf8]{inputenc}
\usepackage[english]{babel}
\usepackage{amsmath,amssymb,amsfonts,amsthm,bm,dsfont}
\usepackage{graphicx}
\usepackage[export]{adjustbox}
\usepackage{xcolor}
\usepackage{braket}
\usepackage{physics}
\usepackage{array,multirow,dcolumn}
\usepackage{algorithm}
\usepackage{algpseudocode}
\usepackage{tikz}
\usepackage[colorlinks,citecolor=blue,linkcolor=blue,urlcolor=blue]{hyperref}

\theoremstyle{plain}
\newtheorem{theorem}{Theorem}
\newtheorem{proposition}[theorem]{Proposition}

\theoremstyle{definition}

\theoremstyle{remark}

\begin{document}

\title{Fibonacci many-body scars in a decorated Rule-54 quantum cellular automaton}

\author{Han-Ze Li}
\affiliation{Institute for Quantum Science and Technology, Shanghai University, Shanghai 200444, China}
\affiliation{Department of Physics, National University of Singapore, Singapore 117542, Singapore}

\author{Jian-Xin Zhong}
\email{jxzhong@shu.edu.cn}
\affiliation{Institute for Quantum Science and Technology, Shanghai University, Shanghai 200444, China}

% \author{Ching Hua Lee}
% \email{phylch@nus.edu.sg}
% \affiliation{Department of Physics, National University of Singapore, Singapore 117542, Singapore}

\begin{abstract}
Quantum many-body scars provide a controlled form of weak ergodicity breaking, in which structured nonthermal eigenstates coexist with a thermalizing many-body spectrum. We introduce a qubit-level route to exact scars based on the intrinsic soliton structure of the Rule-54 quantum cellular automaton. A hard-core dimer sector of Rule 54 supplies an exactly translatable protected skeleton, while local projector-controlled decorations are invisible on this skeleton and nontrivial outside it. The protected dynamics is therefore reducible to finite translation orbits, whose Fourier modes form exact Floquet eigenstates with sub-volume-law entanglement. The number of exact scars grows with Fibonacci combinatorics, whereas their fraction in the full qubit Hilbert space remains exponentially small. Finite-size simulations show Page-like eigenstate entanglement, rapid entanglement growth, fidelity decay, and circular unitary ensemble quasienergy statistics in the decorated complement. This construction demonstrates that exact many-body scars can be engineered from native finite-orbit structures of an interacting reversible automaton, and provides a direct starting point for digital quantum simulation of scarred cellular-automaton dynamics.
\end{abstract}

\maketitle

\section{Introduction}
\label{sec:introduction}
Thermalization in isolated quantum many-body systems is commonly formulated through the eigenstate thermalization hypothesis~\cite{Deutsch1991,Srednicki1994,Rigol2008,DAlessio2016}: typical eigenstates of a chaotic system are locally thermal, highly entangled, and statistically featureless. The same chaotic regime is diagnosed spectrally by random-matrix theory~\cite{BohigasGiannoniSchmit1984,MehtaBook,HaakeBook,Atas2013}. Quantum many-body scars (QMBS)~\cite{Bernien2017,Turner2018NatPhys,Turner2018PRB,HoChoiPichlerLukin2019,ShiraishiMori2017,SerbynAbanin2021,MoudgalyaReview2021,Chandran2023} violate this expectation in a selective way: a sparse set of structured nonthermal eigenstates produces revivals and slow relaxation from special initial states, without making the full spectrum nonthermal. This weak form of ergodicity breaking is distinct from integrable dynamics~\cite{Bethe1931,CauxEssler2013}, many-body localization~\cite{Anderson1958,Basko2006,NandkishoreHuse2015,Abanin2019}, and Hilbert-space fragmentation~\cite{Sala2020,Khemani2020Fragmentation,MoudgalyaReview2021}, where nonthermal structure is organized by extensive conservation laws, localization, or disconnected Krylov sectors. The inverse problem addressed here is therefore precise: can one engineer a local Floquet dynamics in which an exactly solvable nonthermal skeleton is protected, while the rest of the Hilbert space remains dynamically generic?
On the other hand, further developments of the QMBS program include emergent algebraic mechanisms, quasiparticle and mixed-phase-space interpretations, exact scar constructions, and quantum-computing or data-driven detection protocols~\cite{Choi2019,Khemani2019,LinMotrunich2020,Michailidis2020,Yao2021ScarsCriticality,BurkeDooley2025Temperature,FengZhang2024QCNNScars,Gustafson2023Preparing}.

Quantum cellular automata (QCA)~\cite{Watrous1995,SchumacherWerner2004,GrossNesmeVogtsWerner2012,Arrighi2019,Farrelly2020} are a natural language for this problem. A QCA is a reversible, strictly local, discrete-time update rule, and hence already has the microscopic form of a digital quantum circuit~\cite{Lloyd1996,Raussendorf2005,Preskill2018,Qiskit2021,SmithKimPollmannKnolle2019,Arute2019}. Compared with random unitary circuits~\cite{NahumRuhmanVijayHaah2017,vonKeyserlingk2018,JonayHuseNahum2018,ZhouNahum2019}, QCA can contain deterministic structures such as finite orbits, solitons, local constraints, and exact tensor-network identities. The Rule-54 QCA~\cite{Bobenko1993,ProsenMejiaMonasterio2016,ProsenBuca2017,BucaGarrahanProsenVanicat2019,Buca2021Review,KlobasBertiniPiroli2021,KlobasBertini2021,LopezPiqueres2022,PalettaProsen2026} is especially suited to this aim. It is a qubit QCA generated by a reversible three-site update, supports soliton-like excitations with interaction-induced scattering shifts, and provides an exactly solvable benchmark for interacting relaxation and entanglement growth. It is therefore neither a featureless random circuit nor a trivial permutation model: it has enough deterministic structure to support protected finite-orbit dynamics, while still representing an interacting many-body automaton.

Our construction uses this Rule-54 soliton structure as the scar skeleton. Projector embedding~\cite{ShiraishiMori2017,Moudgalya2018,SchecterIadecola2019,MarkLinMotrunich2020,Pakrouski2020} gives a general method for making selected states invisible to parts of a local dynamics. In dual-unitary circuits~\cite{BertiniKosProsen2018,BertiniKosProsen2019PRX,BertiniKosProsen2019PRL,PiroliBertiniCirac2020,Bertini2024SpaceLikeDynamics,Logaric2024DUScars}, related ideas can embed scarred subspaces whose protected motion is essentially SWAP dynamics. Rule 54 offers a different microscopic route. Its elementary gate is a three-site cellular-automaton update, and the natural protected motion is not SWAP motion but native soliton translation. We identify a hard-core dimer sector in which one Rule-54 Floquet period acts as a rigid translation on a coarse-grained lattice. We then insert local projector-decorated gates after the Rule-54 half steps. These decorations are identity operators on the local patterns visited by the protected trajectory and nontrivial on forbidden local patterns, so the protected soliton motion survives exactly while the complement is locally decorated.

This produces a qubit-only family of exact Floquet scars. The protected configurations are hard-core dimer configurations on a cycle, so their number grows with Fibonacci combinatorics, whereas their fraction in the full qubit Hilbert space is exponentially small. The protected translation decomposes this constrained space into finite orbits, and Fourier modes on these orbits give exact eigenstates of the full decorated Floquet unitary. The resulting orbit-Fourier scars have sub-volume-law entanglement, bounded by the orbit length, in contrast with the Page entanglement expected for generic Floquet eigenstates~\cite{Page1993,CalabreseCardy2005,AlbaCalabrese2017,NahumRuhmanVijayHaah2017,JonayHuseNahum2018}. Numerically, site-dependent non-diagonal decorations produce high-entanglement bulk eigenstates, rapid entanglement growth, fidelity decay, and quasienergy statistics consistent with the circular uniary ensemble (CUE) symmetry class after removing the protected sector. Thus the construction converts a native finite-orbit structure of an interacting reversible automaton into an exact scar sector embedded in a thermal-like Floquet background.

The circuit formulation also makes the model relevant for digital quantum simulation~\cite{Lloyd1996,Raussendorf2005,Preskill2018,Qiskit2021,SmithKimPollmannKnolle2019,Arute2019,KohTaiLee2022PRL,KohTaiLee2022npj,KohTaiLee2024NatComm,ShenChenYangLee2025NatComm,KohXueTaiKohLee2025,ShenLee2026,ShenChenLee2025,ShenHaoLee2025,NgWangChanShenChenLee2026}. The local Rule-54 update can be decomposed into elementary reversible gates, and hardware-friendly diagonal decorations can be implemented as pattern-controlled phase gates on forbidden local configurations. Beyond the present construction, decorated Rule-54 circuits provide a controlled setting for monitored entanglement dynamics~\cite{LiChenFisher2018,SkinnerRuhmanNahum2019,ChanNandkishorePretkoSmith2019,BaoChoiAltman2020,IppolitiKhemani2021}, quantum Mpemba effects~\cite{AresMurcianoCalabrese2023,LiuZhangYinZhang2024,LiuZhangYinZhangYao2023,LiuZhangYinZhangYao2025,YuLiuShiXinZhang2025,XuEtAl2025Mpemba,YuLiZhang2025CPL,LiLeeLiuShiXinZhangZhong2025,ZhangLiHuangZhaoZhong2026}, and nonstabilizerness or magic spreading~\cite{BravyiKitaev2005,Veitch2014,HowardCampbell2017,SeddonCampbell2019,LeoneOlivieroHamma2022,HaugPiroli2023,Iannotti2025,ViscardiDalmonteHammaTirrito2025,Collura2026,Korbany2025Magic,QianWang2025,HuangLiLeeZhong2025}. Sec.~\ref{sec:rule54_qca} defines the Rule-54 QCA and the hard-core soliton sector. Sec.~\ref{sec:projector_decorated_rule54} introduces the projector-decorated circuit. Sec.~\ref{sec:fibonacci_scar_sector} constructs the orbit-Fourier scar eigenstates and derives their counting and entanglement properties. Sec.~\ref{sec:numerics} presents numerical diagnostics, Sec.~\ref{sec:digital_realization} discusses digital realization, and Sec.~\ref{sec:conclusion} concludes.

\section{Model}
\subsection{Rule-54 quantum cellular automaton}
\label{sec:rule54_qca}

We consider a periodic chain of $N\!=\!2L$ qubits. The local Hilbert space at site $j$ is $\mathcal H_j\!\simeq\!\mathbb C^2$, with computational basis $|s_j\rangle_j$ and $s_j\!\in\!\{0,1\}$. The full Hilbert space is $\mathcal H_N\!=\!\bigotimes_{j=0}^{N-1}\mathcal H_j$ and all site labels are understood modulo $N$.
The local Rule-54 update is a three-site reversible gate $R_j$ acting on the triplet $(j-1,\,j,\,j+1)$. It copies the two neighboring qubits and updates only the central qubit:
\begin{align}
R_j|x,y,z\rangle_{j-1,j,j+1}
=
|x,\chi(x,y,z),z\rangle_{j-1,j,j+1},
\label{eq:Rj_action_main}
\end{align}
where
\begin{align}
\chi(x,y,z)=(x+y+z+xz)\bmod2.
\label{eq:chi_rule54_main}
\end{align}
For reference, the central update has the truth table
\begin{align}
\begin{array}{ccccccccc}
x & 0 & 0 & 0 & 0 & 1 & 1 & 1 & 1 \\
y & 0 & 0 & 1 & 1 & 0 & 0 & 1 & 1 \\
z & 0 & 1 & 0 & 1 & 0 & 1 & 0 & 1 \\
\hline
\chi(x,y,z) & 0 & 1 & 1 & 0 & 1 & 1 & 0 & 0
\end{array}
\label{eq:rule54_table_compact}
\end{align}
The corresponding local matrix elements are
\begin{align}
\langle x',y',z'|R_j|x,y,z\rangle
=
\delta_{x',x}\delta_{z',z}\delta_{y',\chi(x,y,z)}.
\label{eq:Rj_matrix_main}
\end{align}
Since $\chi(x,\chi(x,y,z),z)\!=\!y$, each local update is an involution, $R_j^2\!=\!\mathbb I$, and hence a permutation unitary. This deterministic structure is important: finite orbits of the computational-basis permutation generated by Rule 54 can be diagonalized exactly by discrete Fourier transform.

The Floquet update is built from even- and odd-center layers. In this paper we use the convention
\begin{align}
U_e=\prod_{r=0}^{L-1}R_{2r},
\quad
U_o=\prod_{r=0}^{L-1}R_{2r+1},
\quad
U_{54} \equiv U_oU_e.
\label{eq:Ue_Uo_main}
\end{align}
The opposite convention only reverses the direction of the coarse-grained translation below. Within each layer the updates commute: gates with the same center parity update distinct central sites, while shared neighboring sites are copied through as controls. These are also illustrated in Fig.~\ref{fig:0}(a).

We now identify the hard-core soliton sector. Introduce a coarse-grained lattice of length $L$, whose $r$-th site corresponds to the physical even bond $(2r,\,2r+1)$. A coarse-grained configuration is $\mathbf a\!=\!(a_0,\ldots,a_{L-1})$ with $a_r\!\in\!\{0,1\}$. We impose the hard-core constraint
\begin{align}
\mathcal I_L=
\left\{
\mathbf a\in\{0,1\}^L:
 a_ra_{r+1}=0\ \text{for all }r
\right\},
\label{eq:IL_main}
\end{align}
where $r$ is understood modulo $L$. The corresponding even-dimer and odd-dimer product states are
\begin{align}
\begin{split}
|\mathbf a\rangle_e
&=
\bigotimes_{r=0}^{L-1}|a_r\rangle_{2r}|a_r\rangle_{2r+1},
\\
|\mathbf a\rangle_o
&=
\bigotimes_{r=0}^{L-1}|a_r\rangle_{2r+1}|a_r\rangle_{2r+2}.
\label{eq:even_odd_dimer_main}
\end{split}
\end{align}
Thus $a_r\!=\!1$ represents a physical $11$ dimer on the even bond $(2r,\,2r+1)$, and $a_r\!=\!0$ represents a $00$ dimer. The hard-core condition forbids neighboring occupied dimers on the coarse-grained cycle.

Then, let $T$ be the coarse-grained translation $(T\mathbf a)_r\!=\!a_{r-1}$. The key Rule-54 identity is
\begin{align}
U_e|\mathbf a\rangle_e=|\mathbf a\rangle_o,
\quad
U_o|\mathbf a\rangle_o=|T\mathbf a\rangle_e,
\quad
U_{54}|\mathbf a\rangle_e=|T\mathbf a\rangle_e.
\label{eq:rule54_translation_main}
\end{align}
In $|\mathbf a\rangle_e$, the three qubits around the even center $2r$ are $(a_{r-1},\,a_r,\,a_r)$. The updated qubit is $\chi(a_{r-1},\,a_r,\,a_r)\!=\!(a_{r-1}+a_{r-1}a_r)\bmod2\!=\!a_{r-1}$, where the last step uses $a_{r-1}a_r=0$. Since the odd site $2r+1$ is unchanged by $U_e$, the output has qubits $(s'_{2r},\,s'_{2r+1})\!=\!(a_{r-1},\,a_r)$, which is exactly $|\mathbf a\rangle_o$. The second identity follows by the same calculation on the odd centers. A fully expanded proof, including the periodic-boundary convention, is given in Appendix~\ref{app:rule54_translation}.

We define the even hard-core soliton subspace
\begin{align}
\mathcal T_e={\rm span}\{|\mathbf a\rangle_e:\mathbf a\in\mathcal I_L\}.
\label{eq:Te_main}
\end{align}
Eq.~\eqref{eq:rule54_translation_main} implies that $\mathcal T_e$ is invariant under $U_{54}$ and that the restricted dynamics is exactly translation:
\begin{align}
U_{54}|_{\mathcal T_e}=T|_{\mathcal T_e}.
\label{eq:U54_equals_T_main}
\end{align}
For example, the one-soliton state $|1100\cdots0\rangle$ evolves as $|1100\cdots0\rangle\!\to\!|001100\cdots0\rangle\!\to\!|00001100\cdots0\rangle\!\to\!\cdots$, with period $L$ on a periodic chain. This rigid soliton translation is the protected dynamics that will survive the decoration.

\subsection{Projector-decorated Rule-54 circuit}
\label{sec:projector_decorated_rule54}

\begin{figure}[bt]
\centering
\includegraphics[width=\columnwidth]{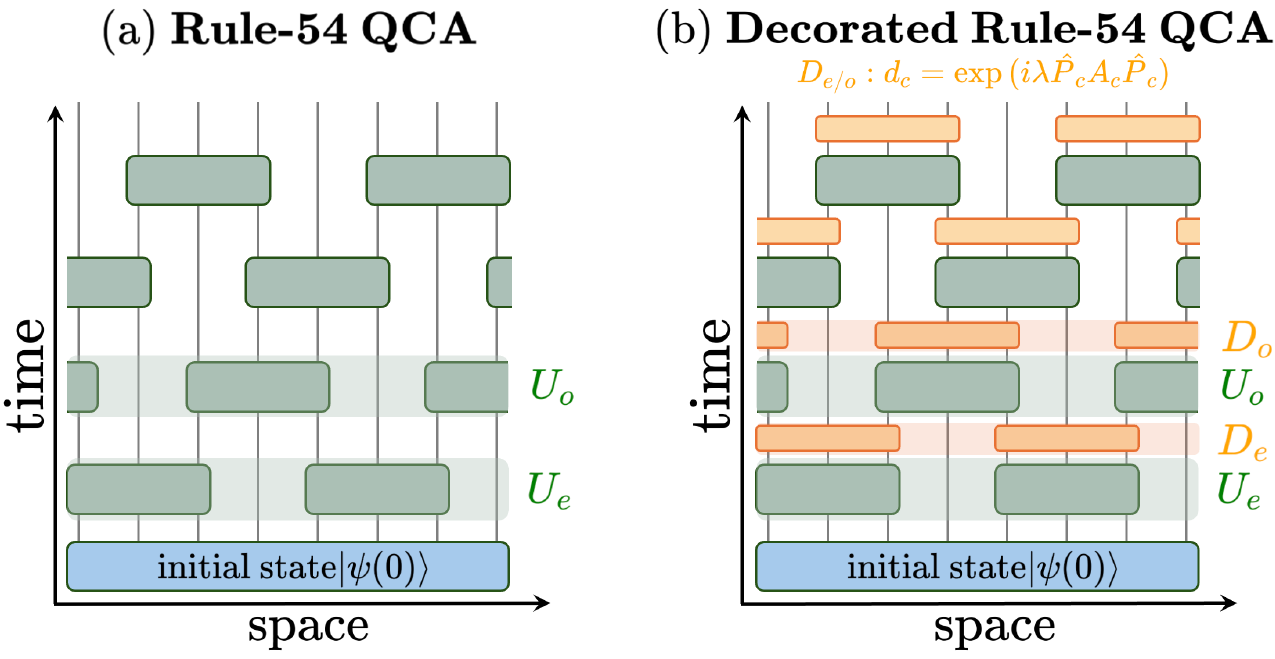} 
\caption{
\emph{Circuit architecture of the bare and decorated Rule-54 quantum cellular automaton.}
(a) Bare Rule-54 QCA. The evolution is generated by alternating layers of three-site reversible Rule-54 gates. 
The green blocks denote local updates $R_j$ arranged in the even and odd layers $U_e$ and $U_o$, and one Floquet period is $U_{54}\!=\!U_oU_e$. 
The vertical direction represents discrete Floquet time, while the horizontal direction labels physical space.
(b) Decorated Rule-54 QCA. The same Rule-54 backbone is supplemented by local projector-protected decoration layers. 
After the even Rule-54 layer $U_e$, a decoration layer $D_e$ is inserted; after the odd Rule-54 layer $U_o$, a decoration layer $D_o$ is inserted. 
The resulting Floquet unitary is therefore $U_{\rm DR54}\!=\!D_oU_oD_eU_e$. 
The decorations are constructed to act trivially on the protected hard-core soliton trajectory, while generating nontrivial local dynamics outside the protected sector.
}
\label{fig:0}
\end{figure}

We now define local decorations that are invisible on the hard-core soliton trajectory but generic outside it. The bare trajectory alternates as $\mathcal T_e\xrightarrow{U_e}\mathcal T_o\xrightarrow{U_o}\mathcal T_e$, where $\mathcal T_o\!=\!{\rm span}\{|\mathbf a\rangle_o:\mathbf a\!\in\!\mathcal I_L\}$. Hence a decoration inserted after $U_e$ must act trivially on $\mathcal T_o$, and a decoration inserted after $U_o$ must act trivially on $\mathcal T_e$.

For a center site $c$, write a local three-qubit pattern as $\tau\!=\!(\tau_L,\tau_C,\tau_R)\!\in\!\{0,1\}^3$ on sites $(c-1,\,c,\,c+1)$. The allowed patterns compatible with the even-dimer hard-core sector are
\begin{align}
X_{2r}^{E}=\{000,100,011\},
\quad
X_{2r+1}^{E}=\{000,001,110\}.
\label{eq:X_E_main}
\end{align}
The allowed patterns compatible with the odd-dimer hard-core sector are shifted by one site:
\begin{align}
X_{2r}^{O}=\{000,001,110\},
\quad
X_{2r+1}^{O}=\{000,100,011\}.
\label{eq:X_O_main}
\end{align}
For $\mu\!=\!E,O$, define the local forbidden-pattern projector
\begin{align}
P_c^\mu=\mathbb I_{c-1,c,c+1}-\sum_{\tau\in X_c^\mu}|\tau\rangle\langle\tau|_{c-1,c,c+1},
\label{eq:local_projector_main}
\end{align}
and let $\hat P_c^\mu$ denote its embedding into the full chain. The corresponding global projectors are
\begin{align}
\Pi_e=\prod_{c=0}^{N-1}(\mathbb I-\hat P_c^E),
\quad
\Pi_o=\prod_{c=0}^{N-1}(\mathbb I-\hat P_c^O).
\label{eq:global_projectors_main}
\end{align}
Their images are $\mathcal T_e$ and $\mathcal T_o$, respectively. The proof is local: the allowed triplets force each even or odd dimer to be either $00$ or $11$, and the missing triplets enforce the hard-core constraint between neighboring dimers. The details are given in Appendix~\ref{app:projector_images}.

Then, let $A_c^{(e)}$ and $A_c^{(o)}$ be arbitrary three-qubit Hermitian operators, possibly chosen independently from site to site. We define projected Hermitian generators
\begin{align}
G_c^{(e)}=\hat P_c^O A_c^{(e)}\hat P_c^O,
\quad
G_c^{(o)}=\hat P_c^E A_c^{(o)}\hat P_c^E.
\label{eq:G_generators_main}
\end{align}
The superscript $(e)$ means that the gate is inserted after the even Rule-54 layer; it uses $P^O$ because the state is then in $\mathcal T_o$. Similarly, the superscript $(o)$ means that the gate is inserted after the odd Rule-54 layer; it uses $P^E$ because the state is then in $\mathcal T_e$.

The local decoration gates and the full decoration layers are
\begin{align}
d_c^{(e)}=\exp(i\lambda_eG_c^{(e)})&,
\quad
d_c^{(o)}=\exp(i\lambda_oG_c^{(o)}),
\\
D_e=\overleftarrow{\prod_{c=0}^{N-1}}d_c^{(e)}&,
\quad
D_o=\overleftarrow{\prod_{c=0}^{N-1}}d_c^{(o)}.
\label{eq:decorations_main}
\end{align}
The ordering in the products is fixed but arbitrary. For a strictly parallel circuit implementation, the same product can be decomposed into a constant number of sublayers by coloring the three-site supports; this locality point is discussed in Appendix~\ref{app:decorated_proof}. The decorated Floquet unitary [see also Fig.~\ref{fig:0}(b)] is
\begin{align}
U_{\rm DR54}=D_oU_oD_eU_e.
\label{eq:UDR54_main}
\end{align}

The reason for the projected form is immediate. If $|\psi_o\rangle\!\in\!\mathcal T_o$, then $\hat P_c^O|\psi_o\rangle\!=\!0$ for all $c$, hence $G_c^{(e)}|\psi_o\rangle\!=\!0$ and $d_c^{(e)}|\psi_o\rangle\!=\!|\psi_o\rangle$. Thus $D_e$ is the identity on $\mathcal T_o$. Similarly, $D_o$ is the identity on $\mathcal T_e$:
\begin{align}
D_e|_{\mathcal T_o}=\mathbb I_{\mathcal T_o},
\quad
D_o|_{\mathcal T_e}=\mathbb I_{\mathcal T_e}.
\label{eq:D_identity_main}
\end{align}
Combining this with Eq.~\eqref{eq:rule54_translation_main}, we obtain the exact protected dynamics
\begin{align}
\begin{split}
U_{\rm DR54}|\mathbf a\rangle_e
&=
D_oU_oD_eU_e|\mathbf a\rangle_e \\
&=
|T\mathbf a\rangle_e,
\label{eq:UDR54_translation_main}
\end{split}
\end{align}
where $\mathbf a\!\in\!\mathcal I_L$.
Equivalently,
\begin{align}
U_{\rm DR54}|_{\mathcal T_e}&=T|_{\mathcal T_e},
\nonumber\\
(\mathbb I-
\Pi_e)&U_{\rm DR54}\Pi_e=0.
\label{eq:no_leakage_main}
\end{align}
The full proof, including the operator identities for the projectors, is given in Appendix~\ref{app:decorated_proof}. And Fig.~\ref{fig:0} summarizes the circuit structure used throughout this work.

\section{Analysis results}
\label{sec:fibonacci_scar_sector}
The previous section established the central structural identity of the decorated circuit: on the even hard-core soliton subspace $\mathcal T_e$, the full Floquet unitary $U_{\rm DR54}$ acts exactly as the coarse-grained translation $T$. We now diagonalize this protected motion. The result is a complete, analytically controlled family of Floquet eigenstates of the full decorated circuit. These states are the scars of the model. They are exponentially numerous, but still occupy an exponentially small fraction of the full qubit Hilbert space; moreover, their entanglement is bounded by a logarithm of system size.

The construction is conceptually simple. The protected basis states $|\mathbf a\rangle_e$ are labelled by hard-core configurations $\mathbf a\!\in\!\mathcal I_L$. Since $U_{\rm DR54}$ translates these labels, the problem of finding eigenstates in $\mathcal T_e$ reduces to diagonalizing a finite permutation. The eigenstates of a finite permutation are Fourier modes on its cycles. In the present setting the cycles are translation orbits of hard-core soliton configurations.

\subsection{Size of the protected sector}
We first count the dimension of $\mathcal T_e$. Since the states $|\mathbf a\rangle_e$ with $\mathbf a\!\in\!\mathcal I_L$ are mutually orthonormal, $\dim\mathcal T_e\!=\!|\mathcal I_L|$. The set $\mathcal I_L$ is the set of independent sets on the cycle graph $C_L$. Equivalently, it is the set of binary strings on a ring with no adjacent $1$'s. This elementary combinatorial fact is already enough to distinguish the protected sector from a conventional local-alphabet subspace: its growth is exponential, but with a base smaller than the full qubit Hilbert-space growth.
\begin{proposition}
\label{thm:fibonacci_count_main}
For a periodic chain of $N$ qubits, the even hard-core soliton sector has dimension $\dim\mathcal T_e\!=\!|\mathcal I_L|\!=\!F_{L-1}+F_{L+1}\!=\!
\varphi^L+(-\varphi^{-1})^L$,
where $F_n$ denotes the Fibonacci number and $\varphi\!=\!\frac{(1+\sqrt 5)}{2}$. Consequently, $\dim\mathcal T_e\!\sim\!\varphi^L$,
\begin{align}
\frac{\dim\mathcal T_e}{2^N}
\!=\!\frac{F_{L-1}+F_{L+1}}{4^L}
\sim
\left(\frac{\sqrt\varphi}{2}\right)^N.
\label{eq:fibonacci_fraction_theorem_main}
\end{align}
\end{proposition}

\begin{proof}
The compatibility matrix for two neighboring hard-core occupations $a,b\!\in\!\{0,1\}$ is
\begin{align}
M_{ab}=1-ab,
\quad
M=
\begin{pmatrix}
1&1\\
1&0
\end{pmatrix}.
\label{eq:hardcore_transfer_main}
\end{align}
A periodic configuration contributes if every adjacent pair is compatible, including the boundary pair. Therefore
\begin{align}
|\mathcal I_L|
&=
\sum_{a_0,\ldots,a_{L-1}=0}^{1}
M_{a_0a_1}M_{a_1a_2}\cdots M_{a_{L-2}a_{L-1}}M_{a_{L-1}a_0}
\nonumber\\
&=
\sum_{a_0=0}^{1}(M^L)_{a_0a_0}
\nonumber\\
&=
{\rm Tr}\,M^L.
\label{eq:trace_count_main}
\end{align}
The eigenvalues of $M$ are $\varphi$ and $-\varphi^{-1}$, which gives $|\mathcal I_L|\!=\!\varphi^L+(-\varphi^{-1})^L$. The equality to $F_{L-1}+F_{L+1}$ follows from Binet's formula. The asymptotic fraction follows from $2^N\!=\!4^L$. A more detailed derivation, including the fixed-soliton-number count, is given in Appendix~\ref{app:fibonacci_counting}.
\end{proof}

The number $F_{L-1}+F_{L+1}$ is also the Fibonacci count of the cycle. We will nevertheless refer to the sector as Fibonacci-dimensional, because its asymptotic growth is governed by the Fibonacci golden ratio. The important physical point is the scaling in Eq.~\eqref{eq:fibonacci_fraction_theorem_main}: the scar sector is exponentially large, but its spectral weight in the full Hilbert space vanishes exponentially. This is the correct scaling for weak ergodicity breaking, rather than strong fragmentation of Hilbert space.

\subsection{Translation orbits and cyclic subspaces}
The translation $T$ partitions $\mathcal I_L$ into disjoint finite orbits. We write this decomposition as $\mathcal I_L\!=\!\bigsqcup_\nu\mathcal O_\nu$, where
\begin{align}
\mathcal O_\nu
=
\left\{
\mathbf a_\nu,
T\mathbf a_\nu,
\ldots,
T^{p_\nu-1}\mathbf a_\nu
\right\},
\quad
T^{p_\nu}\mathbf a_\nu=\mathbf a_\nu.
\label{eq:orbit_definition_main}
\end{align}
Here $p_\nu$ is the minimal period of the representative $\mathbf a_\nu$. Since $T^L\!=\!\mathbb I$ on the length-$L$ coarse-grained ring, every period satisfies $p_\nu|L$ and $p_\nu\!\le\!L$. Each orbit defines an invariant cyclic subspace
\begin{align}
\mathcal H_\nu
=
{\rm span}\left\{
|T^t\mathbf a_\nu\rangle_e:
 t=0,\ldots,p_\nu-1
\right\}
\subset\mathcal T_e.
\label{eq:orbit_subspace_main}
\end{align}
On this subspace, the decorated circuit is simply a cyclic shift,
\begin{align}
U_{\rm DR54}|T^t\mathbf a_\nu\rangle_e
=
|T^{t+1}\mathbf a_\nu\rangle_e,
\label{eq:cyclic_shift_main}
\end{align}
where $t\!=\!0,\ldots,p_\nu-1$, and $t+1$ is understood modulo $p_\nu$. Thus the problem has been reduced from diagonalizing the full $2^N$-dimensional Floquet unitary to diagonalizing independent finite cyclic shifts.

\subsection{Exact scar eigenstates}
The eigenvectors of a cyclic shift are discrete Fourier modes. For each orbit $\mathcal O_\nu$ and each integer $m$, define
\begin{align}
|\Phi_{\nu,m}\rangle
=
\frac{1}{\sqrt{p_\nu}}
\sum_{t=0}^{p_\nu-1}
\exp\left(-\frac{2\pi i mt}{p_\nu}\right)
|T^t\mathbf a_\nu\rangle_e.
\label{eq:orbit_fourier_state_main_new}
\end{align}
These are the exact scar eigenstates: they are exact eigenstates of the full decorated circuit, they are nonthermal, and their fraction among all Floquet eigenstates is exponentially small in system size.

\begin{proposition}
\label{thm:orbit_fourier_main}
For every translation orbit $\mathcal O_\nu$ and every $m\!=\!0,\ldots,p_\nu-1$, the state $|\Phi_{\nu,m}\rangle$ is an exact Floquet eigenstate of $U_{\rm DR54}$ with eigenphase $2\pi m/p_\nu$:
\begin{align}
U_{\rm DR54}|\Phi_{\nu,m}\rangle
=
\exp\left(\frac{2\pi i m}{p_\nu}\right)|\Phi_{\nu,m}\rangle.
\label{eq:orbit_fourier_eigenvalue_main_new}
\end{align}
Moreover, the set of all $|\Phi_{\nu,m}\rangle$ forms an orthonormal basis of $\mathcal T_e$.
\end{proposition}

\begin{proof}
Using Eq.~\eqref{eq:cyclic_shift_main}, we find
\begin{align}
U_{\rm DR54}|\Phi_{\nu,m}\rangle
&=
\frac{1}{\sqrt{p_\nu}}
\sum_{t=0}^{p_\nu-1}
\exp\left(-\frac{2\pi i mt}{p_\nu}\right)
|T^{t+1}\mathbf a_\nu\rangle_e
\nonumber\\
&=
\frac{1}{\sqrt{p_\nu}}
\sum_{t'=0}^{p_\nu-1}
\exp\left[-\frac{2\pi i m(t'-1)}{p_\nu}\right]
|T^{t'}\mathbf a_\nu\rangle_e
\nonumber\\
&=
\exp\left(\frac{2\pi i m}{p_\nu}\right)|\Phi_{\nu,m}\rangle.
\label{eq:main_fourier_proof_chain}
\end{align}
States from different orbits have disjoint computational-basis support. Within a fixed orbit, orthonormality follows from the orthogonality of finite Fourier modes. Since the total number of states is $\sum_\nu p_\nu\!=\!|\mathcal I_L|\!=\!\dim\mathcal T_e$, the orbit-Fourier states form a complete orthonormal basis of $\mathcal T_e$. More details are given in Appendix~\ref{app:orbit_fourier}.
\end{proof}

Eq.~\eqref{eq:orbit_fourier_eigenvalue_main_new} is the main exact spectral result of the paper. The scar eigenphases are fixed by the periods of finite translation orbits. Since distinct orbits can have the same period, and different periods can share rational phases, eigenphase degeneracies are common. A numerical eigensolver may therefore output arbitrary orthonormal combinations inside a degenerate protected eigenspace. The invariant diagnostic is the protected-sector weight $w_\alpha\!=\!\langle\phi_\alpha|\Pi_e|\phi_\alpha\rangle$, which equals one for any eigenvector fully supported in $\mathcal T_e$.
% For example, the vacuum orbit has $p\!=\!1$ and gives the product scar $|0\rangle^{\otimes N}$ with eigenphase zero. The one-soliton orbit generated by $|1100\cdots0\rangle$ has period $L$ and gives $L$ scars with phases $2\pi m/L$. For $N\!=\!10$, one has $L\!=\!5$ and $|\mathcal I_5|\!=\!F_4+F_6\!=\!11$: one vacuum scar, five one-soliton scars, and five two-soliton hard-core scars. This explains why numerical spectra reveal more low-entanglement protected eigenstates than those obtained from a single one-soliton orbit alone.

\subsection{Entanglement bound}
The exact scars are not thermal eigenstates. This is already visible from their wave functions: each one is a coherent superposition of product states belonging to a single finite translation orbit. Since the orbit length is at most $L$, the number of product components is at most $L$. This immediately bounds the Schmidt rank across any bipartition.
\begin{proposition}
\label{prop:entanglement_bound_main}
For any bipartition $A\cup\bar A$ of the physical chain and any orbit-Fourier scar $|\Phi_{\nu,m}\rangle$, the von Neumann entropy satisfies
\begin{align}
S_A(|\Phi_{\nu,m}\rangle)
\le
\log p_\nu
\le
\log L
=
\log(N/2).
\label{eq:entropy_bound_main_new}
\end{align}
Also, the same upper bound holds for all positive R\'enyi entropies.
\end{proposition}

\begin{proof}
Each basis state in the orbit is a computational-basis product state and therefore factorizes across any bipartition as $|T^t\mathbf a_\nu\rangle_e\!=\!|\ell_t\rangle_A\otimes|r_t\rangle_{\bar A}$. Substituting this factorization into Eq.~\eqref{eq:orbit_fourier_state_main_new} gives an expansion of $|\Phi_{\nu,m}\rangle$ as at most $p_\nu$ product terms across the cut. Hence its Schmidt rank is at most $p_\nu$. The entropy of a state with Schmidt rank $\chi$ is bounded by $\log\chi$, giving Eq.~\eqref{eq:entropy_bound_main_new}. Appendix~\ref{app:entanglement_bound} gives the full proofs.
\end{proof}

This logarithmic bound is parametrically smaller than the Page volume-law entanglement expected for generic Floquet eigenstates. It is also the reason why the scar states appear as a low-entanglement band in numerical spectra. The bound should be interpreted as a statement about the analytic orbit-Fourier basis. If a degenerate scar eigenspace is diagonalized numerically, the returned eigenvectors may be arbitrary mixtures of several orbit-Fourier scars and can have larger entanglement. The protected subspace nevertheless always admits the low-entanglement orbit-Fourier basis constructed above.

In a nutshell, the protected sector contains
\begin{align}
\mathcal N_{\rm scar}=F_{L-1}+F_{L+1}\sim\varphi^L
\label{eq:summary_nscar_main}
\end{align}
exact Floquet eigenstates, with eigenphases determined by finite translation orbits and entanglement bounded by $\log(N/2)$. 
This establishes an exponentially large but exponentially sparse family of exact sub-volume-law scars in a qubit Floquet circuit. Fig.~\ref{fig:1} summarizes the two analytic consequences of the construction. 
First, the hard-core constraint produces an exponentially large protected space whose dimension is fixed purely by the combinatorics of independent sets on a cycle. 
The growth Eq.~\eqref{eq:summary_nscar_main} is exponential in $L$, but its fraction in the full qubit Hilbert space is $\mathcal N_{\rm scar}/2^N\!\sim\!(\sqrt{\varphi}/2)^N$, and therefore vanishes exponentially. 
This scaling is the sense in which the construction realizes weak rather than strong ergodicity breaking. Second, the protected translation dynamics gives an explicit scar basis. 
Each translation orbit $\mathcal O_\nu$ contributes $p_\nu$ orbit-Fourier eigenstates with eigenphases $2\pi m/p_\nu$. 
For the representative case $N\!=\!10$, the protected sector contains one vacuum scar and two nontrivial period-five hard-core soliton orbits, giving $F_4+F_6\!=\!11$ exact scars in total. 
The entanglement data in Fig.~\ref{fig:1}(b) are obtained directly from the analytic wave functions in Eq.~\eqref{eq:orbit_fourier_state_main_new}. 
All points lie below the bound $S_{N/2}\!\leq\!\log L$, confirming that the exact scar basis is sub-volume-law even though the number of scars grows exponentially.

\begin{figure}[bt]
\centering
\includegraphics[width=\columnwidth]{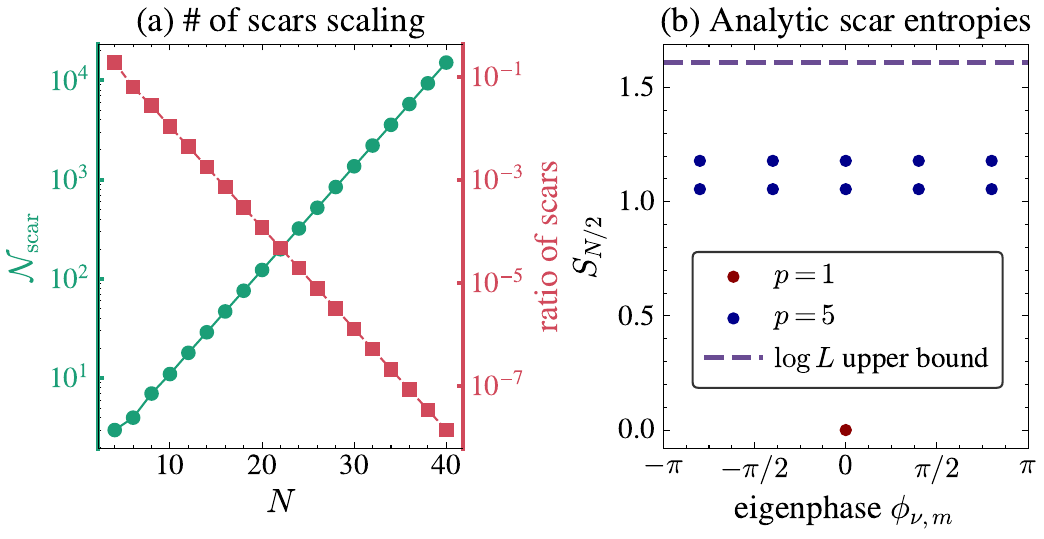} 
\caption{
\emph{Analytic structure of the Fibonacci-dimensional scar sector}.
(a) Number of exact scar eigenstates as a function of system size. 
For a chain of $N=2L$ qubits, the protected hard-core soliton sector has dimension $\mathcal N_{\rm scar}=F_{L-1}+F_{L+1}$, the Fibonacci count of hard-core configurations on a cycle of length $L$. 
The exponential growth $\mathcal N_{\rm scar}\!\sim\!\varphi^L$ is shown together with the exponentially small Hilbert-space fraction $\mathcal N_{\rm scar}/2^N$ when plotted on the secondary axis.
(b) Half-chain entanglement entropies of the analytically constructed orbit-Fourier scar eigenstates for $N\!=\!10$, grouped by their translation-orbit period $p_\nu$. 
The vacuum orbit has $p_\nu\!=\!1$ and gives the product scar $|0\rangle^{\otimes N}$, while the nontrivial hard-core soliton orbits have period $p_\nu\!=\!5$. 
The dashed line marks the rigorous upper bound $S_{N/2}\!\leq\!\log L$, illustrating the sub-volume-law entanglement of the scar basis.
}
\label{fig:1}
\end{figure}

\section{Numerical results}
\label{sec:numerics}

The numerical results in this section show how the exact protected sector appears in finite decorated Rule-54 circuits and how the remaining states behave once site-dependent non-diagonal decorations are included. We focus on three diagnostics: eigenstate entanglement, real-time entanglement and fidelity dynamics, and bulk quasienergy statistics after removing the protected sector. Together these diagnostics illustrate the coexistence of exactly solvable orbit-Fourier scars with a thermal-like decorated Floquet background.

\begin{figure*}[bt]
\centering
\includegraphics[width=\textwidth]{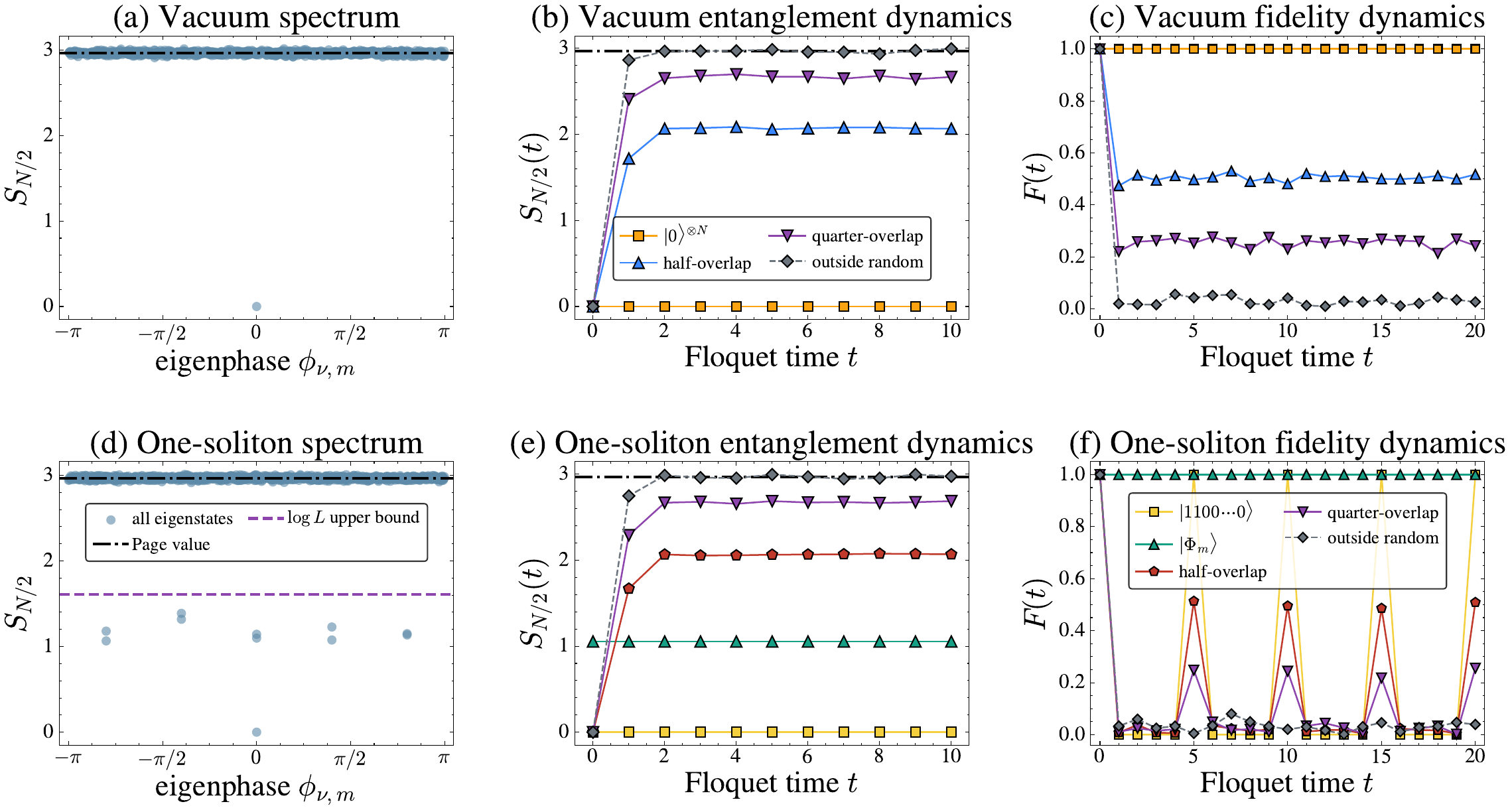} 
\caption{
\emph{Numerical signatures of exact scars in representative projector-decorated Rule-54 circuits.}
The top row shows a construction protecting the vacuum orbit, while the bottom row shows a construction protecting the one-soliton translation orbit generated by $|1100\cdots0\rangle$ for $N\!=\!10$.
(a,d) Half-chain entanglement entropy $S_{N/2}$ of Floquet eigenstates versus quasienergy eigenphase $\phi$. 
The bulk eigenstates concentrate near the finite-size Page value, whereas the protected scar eigenstates appear as low-entanglement outliers. 
In the one-soliton case, the dashed $\log p$ line gives the orbit-length entanglement bound for the analytic orbit-Fourier scars.
(b,e) Entanglement dynamics $S_{N/2}(t)$ from selected initial states. 
The vacuum scar remains a product state, while the one-soliton orbit product state remains inside the protected orbit and therefore has zero entanglement growth. 
The orbit-Fourier scar is an exact eigenstate and has time-independent entropy. 
Initial states with partial scar overlap show intermediate behavior, while outside product states rapidly approach the Page value.
(c,f) Fidelity dynamics $F(t)\!=\!|\langle\psi(0)|\psi(t)\rangle|$. 
The vacuum scar and the orbit-Fourier scar have frozen fidelity, whereas the one-soliton product state exhibits exact revivals with period $p=L$. 
States outside the protected orbit lose fidelity rapidly, and partial-overlap states retain only the component supported on the protected sector.
}
\label{fig:2}
\end{figure*}

Throughout this section the half-chain entanglement entropy is defined as $S_{N/2}(|\psi\rangle)\!=\!-{\rm Tr}\,\rho_A\log\rho_A$, where $\rho_A\!=\!{\rm Tr}_{\bar A}|\psi\rangle\langle\psi|$ and $A$ contains $N/2$ consecutive physical qubits. 
For comparison we use the finite-size Page value for a bipartition with dimensions $d_A\!\le\!d_B$,
\begin{align}
S_{\rm Page}
=
\sum_{k=d_B+1}^{d_Ad_B}\frac{1}{k}
-
\frac{d_A-1}{2d_B}.
\label{eq:numerical_page_value}
\end{align}
For equal halves of an $N$-qubit chain, $d_A\!=\!d_B\!=\!2^{L}$. 
The return fidelity of an initial state $|\psi(0)\rangle$ is
\begin{align}
F(t)
=
|\langle\psi(0)|U_{\rm DR54}^t|\psi(0)\rangle|.
\label{eq:numerical_fidelity}
\end{align}

The decorations are generated by sampling independent complex random matrices on three-qubit supports, symmetrizing them into Hermitian matrices $A_c^{(e/o)}$, and projecting them as in Eq.~\eqref{eq:G_generators_main}. 
The local gates are then exponentiated to obtain $\exp(i\lambda P A P)$. 
This choice preserves the protected path exactly but makes the complement non-diagonal and site dependent. 
In the calculations below, we use exact diagonalization at $N=10$ for the scar spectrum and dynamics diagnostics, and $N=8$ for the bulk quasienergy-spacing statistics after removing the protected sector. 

\subsection{Eigenstate entanglement and dynamics}
Fig.~\ref{fig:2} shows two representative projector-decorated Rule-54 circuits. 
The top row protects the vacuum orbit, whose only basis state is $|0\rangle^{\otimes N}$. 
The bottom row protects the one-soliton translation orbit generated by the product state $|1100\cdots0\rangle$. 
These two examples isolate the two elementary dynamical mechanisms present in the general Fibonacci construction: a fixed product scar and a nontrivial finite translation orbit. 
They should be viewed as representative target-orbit realizations of the general mechanism developed in Sec.~\ref{sec:fibonacci_scar_sector}.

Figs.~\ref{fig:2}(a) and \ref{fig:2}(d) show the half-chain entanglement entropy of Floquet eigenstates as a function of quasienergy eigenphase. 
In both cases, most eigenstates form a high-entanglement band close to the finite-size Page value. 
The protected states appear as low-entanglement outliers. 
For the vacuum construction, the protected state is the product eigenstate $|0\rangle^{\otimes N}$ with zero entanglement. 
For the one-soliton construction, the protected orbit has period $p\!=\!L$; hence the analytic orbit-Fourier scars satisfy \eqref{eq:entropy_bound_main_new}.
The dashed $\log p$ line in Fig.~\ref{fig:2}(d) marks this analytic bound. 
The low-entanglement protected eigenstates lie parametrically below the Page bulk, in agreement with the sub-volume-law bound derived in Proposition~\ref{prop:entanglement_bound_main}.

Figs.~\ref{fig:2}(b) and \ref{fig:2}(e) show entanglement dynamics from selected product or partially overlapping initial states. 
In the vacuum construction, the initial state $|0\rangle^{\otimes N}$ remains an exact product eigenstate and hence has $S_{N/2}(t)\!=\!0$ for all times. 
Initial states with partial overlap with the vacuum scar, labelled as half-overlap and quarter-overlap in the figure, retain a nonthermal component but also contain components outside the protected sector; their entanglement therefore grows to an intermediate value. 
A product state chosen outside the protected path rapidly approaches the Page value.

The one-soliton construction separates two different protected objects. 
The orbit product state $|1100\cdots0\rangle$ is not a Floquet eigenstate. 
Instead, it is translated around the protected orbit: $|1100\cdots0\rangle
\!\to\!
|001100\cdots0\rangle
\!\to\!
|00001100\cdots0\rangle
\!\to\!
\cdots$.
Every state along this trajectory is a computational-basis product state, so its entanglement remains zero. 
By contrast, the orbit-Fourier scar $|\Phi_m\rangle$ is an exact Floquet eigenstate. 
It is generally not a product state at $t=0$ because it is a coherent superposition of several translated product configurations; consequently its entropy is nonzero but time independent. 

Figs.~\ref{fig:2}(c) and \ref{fig:2}(f) show the corresponding fidelity dynamics. 
For the vacuum scar, the fidelity is frozen at unity, $F(t)\!=\!1$ and $|\psi(0)\rangle\!=\!|0\rangle^{\otimes N}$.
The one-soliton orbit-Fourier scar also has frozen fidelity, since time evolution only produces an overall Floquet phase. 
The one-soliton orbit product state instead has exact revivals with period $p\!=\!L$, with $F(t)\!=\!1$ for $t\!=\!0,L,2L,\ldots$, and small overlap with the initial product state at intermediate times. 
This is the finite-orbit version of scarred dynamics: the initial state does not thermalize, but it is not stationary unless it is first diagonalized into the orbit-Fourier basis. 
States outside the protected orbit lose fidelity rapidly, while partial-overlap states inherit a residual or revival component from their projection onto the protected sector.

\subsection{Bulk quasienergy statistics}
\begin{figure}[bt]
\centering
\includegraphics[width=\columnwidth]{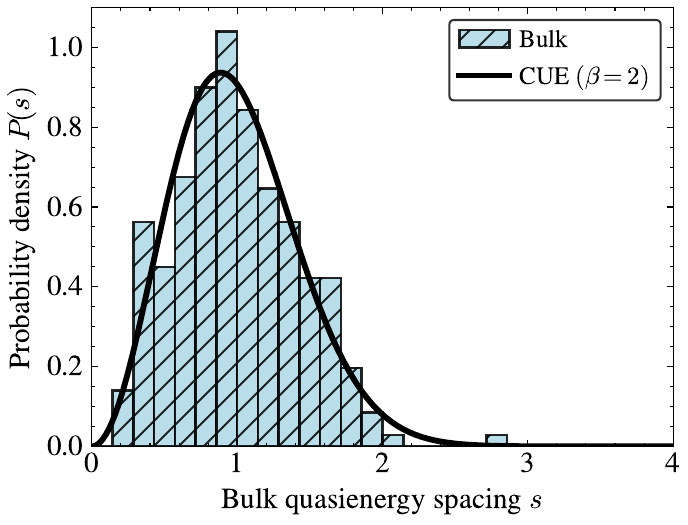} 
\caption{
\emph{Bulk quasienergy spacing statistics after removing the exact scar sector.}
We diagonalize the decorated Floquet unitary, remove the protected scar subspace, and compute the normalized nearest-neighbor quasienergy spacing 
$s\!=\!\Delta\theta/\langle\Delta\theta\rangle$ from the remaining bulk eigenphases on the unit circle. 
The histogram is compared with the $\beta\!=\!2$ random-matrix prediction appropriate for a Floquet unitary without a time-reversal-like antiunitary symmetry, denoted here as the CUE. 
}
\label{fig:3}
\end{figure}
The previous diagnostics demonstrate the protected nonthermal sector and its dynamical consequences. We next examine the decorated complement by computing the quasienergy statistics of the bulk block. After removing the exact protected sector, the remaining spectrum exhibits level repulsion and is compared with the circular-unitary random-matrix benchmark.

For the quasienergy-spacing analysis in Fig.~\ref{fig:3}, the exact scar sector is removed before computing level spacings. 
Equivalently, one works with the bulk block of the unitary acting on the orthogonal complement of the protected sector. 
Let $\{\theta_\alpha\}$ denote the sorted bulk quasienergies on the unit circle. 
We define the circular spacings
\begin{align}
\Delta\theta_\alpha
=
\theta_{\alpha+1}-\theta_\alpha,
\quad
\theta_{\alpha+\mathcal N_{\rm bulk}}\equiv \theta_\alpha+2\pi,
\label{eq:numerical_circular_spacing}
\end{align}
and normalize them by their mean,
\begin{align}
s_\alpha
=
\frac{\Delta\theta_\alpha}{\langle\Delta\theta\rangle}.
\label{eq:numerical_normalized_spacing}
\end{align}
The resulting histogram is compared with the $\beta\!=\!2$ random-matrix spacing distribution,
\begin{align}
P_{\beta=2}(s)
=
\frac{32}{\pi^2}s^2\exp\left(-\frac{4s^2}{\pi}\right),
\label{eq:numerical_cue_spacing}
\end{align}
which is the Wigner-surmise form for the unitary symmetry class. 
For a Floquet unitary without a time-reversal-like antiunitary symmetry, the appropriate circular-ensemble reference is CUE; its local spacing statistics are in the same $\beta\!=\!2$ universality class as GUE. 
This is the relevant comparison for the present simulations because the projected decorations are generated from generic complex Hermitian matrices and therefore do not impose a real or time-reversal-symmetric structure. Fig.~\ref{fig:3} shows that the bulk spacing distribution has level repulsion and is broadly consistent with the CUE curve. 

Taken together, Figs.~\ref{fig:2} and \ref{fig:3} support the intended separation between the two parts of the construction. 
Inside the protected sector, the behavior is exactly solvable and follows the orbit-Fourier theory of Sec.~\ref{sec:fibonacci_scar_sector}. 
Outside the protected sector, the site-dependent projected decorations produce a high-entanglement spectrum, rapid entanglement growth from generic product states, fidelity decay, and bulk quasienergy statistics compatible with the $\beta=2$ random-matrix class. 
The numerical results therefore complement the analytic construction by showing that the exact Fibonacci scar mechanism can coexist with a thermal-like Floquet background.

\section{Digital quantum circuit realization}
\label{sec:digital_realization}
This section gives a concrete implementation protocol. 
The essential point is that both ingredients of the model are local: the bare Rule-54 update is a reversible three-qubit Boolean gate, while the decoration is a projected three-qubit unitary that is designed to be the identity on the protected local patterns.

The local Rule-54 update is the three-qubit reversible gate in Eqs.~\eqref{eq:Rj_action_main} and~\eqref{eq:chi_rule54_main}, equivalently $y'\!=\!y\oplus x\oplus z\oplus xz$.
This gate can be compiled without ancillas using two CNOT gates and one Toffoli gate: ${\rm CNOT}_{x\to y}$, ${\rm CNOT}_{z\to y}$ and ${\rm CCX}_{x,z\to y}$.
The first CNOT implements $y\mapsto y\oplus x$, the second implements $y\mapsto y\oplus x\oplus z$, and the Toffoli implements the final nonlinear term $xz$. 
Thus the neighboring qubits are copied through as controls, while only the central qubit is updated. 
The even and odd Rule-54 layers $U_e$ and $U_o$ are obtained by applying these local gates to all even and odd centers, respectively.
The decoration gates inserted after the Rule-54 half-steps are those defined in Eq.~\eqref{eq:decorations_main}.
Here $A_c^{(e/o)}$ is a three-qubit Hermitian operator and $\hat P_c^{O/E}$ projects onto the local forbidden-pattern subspace. 
Therefore $d_c^{(e/o)}$ is a three-qubit unitary which acts nontrivially only on local patterns that are not compatible with the protected hard-core sector at that stage of the evolution. 
On the allowed local patterns, the projector annihilates the state and the gate reduces to the identity. 
This is the circuit-level version of the exact identities in Eq.~\eqref{eq:D_identity_main}.

For a first digital implementation, a hardware-friendlier choice is the diagonal projected decoration
\begin{align}
d_{c,\mathrm{diag}}^\mu
=
\sum_{\tau\in X_c^\mu}
|\tau\rangle\langle\tau|
+
\sum_{\tau\notin X_c^\mu}
e^{i\theta_{c,\tau}}
|\tau\rangle\langle\tau|,
\quad
\mu=E,O.
\label{eq:digital_diagonal_decoration}
\end{align}
Equivalently,
\begin{align}
d_{c,\mathrm{diag}}^\mu
=
\prod_{\tau\notin X_c^\mu}
C_\tau(\theta_{c,\tau}),
\quad
C_\tau(\theta)
=
\mathbb I+
(e^{i\theta}-1)|\tau\rangle\langle\tau|.
\label{eq:digital_pattern_phase}
\end{align}
Each factor $C_\tau(\theta)$ is a pattern-controlled phase gate. 
It can be implemented by flipping those qubits for which $\tau_i\!=\!0$, applying a multi-controlled phase on the all-one pattern, and then undoing the flips. 
This diagonal construction is a restricted case of the projector decoration. 
It preserves the exact scar sector for the same reason as Eq.~\eqref{eq:D_identity_main}: every allowed local pattern receives unit phase, while only forbidden patterns acquire nontrivial phases.

One decorated Floquet period is then implemented according to Eq.~\eqref{eq:UDR54_main}.
The ordering follows the protected trajectory. 
After $U_e$, a state initially in $\mathcal T_e$ lies in $\mathcal T_o$, so the inserted decoration $D_e$ must be built from the odd-sector forbidden projectors and acts as the identity on $\mathcal T_o$. 
After $U_o$, the state returns to $\mathcal T_e$, so $D_o$ is built from the even-sector forbidden projectors and acts as the identity on $\mathcal T_e$. 
Consequently, the ideal digital circuit obeys the same protected action as Eq.~\eqref{eq:no_leakage_main}.

The most direct measurable signature is the one-soliton translation orbit. 
Prepare the product state
\begin{align}
|\psi(0)\rangle
=
|1100\cdots0\rangle .
\label{eq:digital_one_soliton_state}
\end{align}
In the ideal protected circuit, this state evolves as $|1100\cdots0\rangle
\!\to\!
|001100\cdots0\rangle
\!\to\!
|00001100\cdots0\rangle
\!\to\!
\cdots,$
with period $L$. 
This dynamics can be read out directly from computational-basis samples. 
If $x_t$ denotes the bit string predicted by the protected orbit after $t$ Floquet periods, then the tracking probability $P_{\rm track}(t)\!=\!\Pr[x_t]$,
equals one in the ideal circuit. 
The return probability $P_{\rm ret}(t)\!=\!\Pr[x_0]$,
revives at $t\!=\!0,L,2L,\ldots$. 
A more robust leakage diagnostic is the orbit probability $P_{\mathcal O}(t)\!=\!
\sum_{x\in\mathcal O}\Pr[x]$,
which remains one for any state evolving strictly inside the protected one-soliton orbit.

\begin{figure}[bt]
\centering
\includegraphics[width=0.5\textwidth]{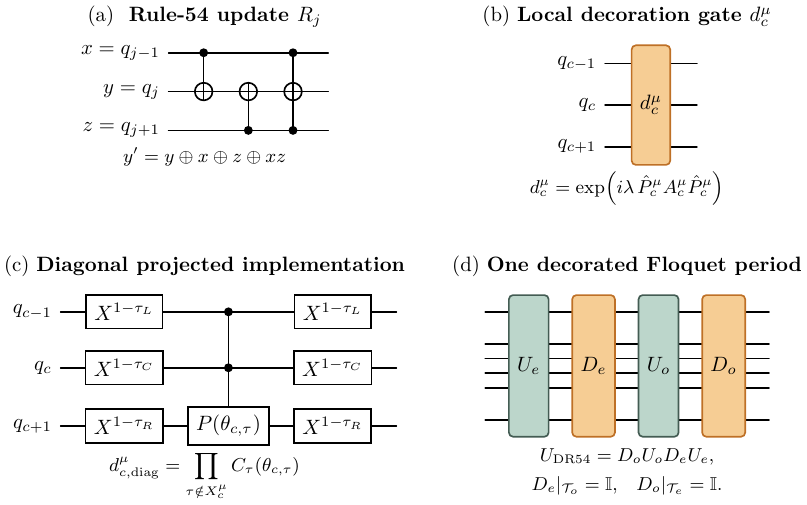} 
\caption{
\emph{Digital-circuit realization of the projector-decorated Rule-54 QCA.}
(a) The local Rule-54 update $R_j$ is implemented by two CNOT gates and one Toffoli gate, giving $y'\!=\!y\oplus x\oplus z\oplus xz$.
(b) The most general local decoration is a projected three-qubit unitary $d_c^\mu$.
(c) A hardware-friendly diagonal realization decomposes $d_c^\mu$ into pattern-controlled phase gates $C_\tau(\theta_{c,\tau})$ acting only on forbidden local patterns $\tau\notin X_c^\mu$. The $X$ gates convert the selected pattern $\tau$ to the all-one pattern before the controlled phase is applied and are then undone.
(d) One decorated Floquet period is $U_{\rm DR54}$. The decoration layers act trivially on the protected sectors reached after the corresponding Rule-54 half-steps.
}
\label{fig:digital_quantikz}
\end{figure}

One may also monitor the spatial motion of the soliton through the dimer-density profile. 
With $n_j\!=\!\frac{1-Z_j}{2}$ with $d_r\!=\!n_{2r}n_{2r+1}$,
the expectation value $\langle d_r(t)\rangle$ should form a translating peak on the coarse-grained lattice. 
This observable can be reconstructed from computational-basis bit strings and does not require full-state tomography. 
A coherent orbit-Fourier scar, Eq.~\eqref{eq:orbit_fourier_state_main_new},
would have frozen fidelity in the ideal circuit, but preparing this superposition is more demanding than preparing the orbit product state. 
Thus the most realistic first digital test is to measure $P_{\rm track}(t)$, $P_{\rm ret}(t)$, $P_{\mathcal O}(t)$, and $\langle d_r(t)\rangle$ for protected product states and compare them with generic outside product states.

Fig.~\ref{fig:digital_quantikz} clarifies how the abstract projector decoration can be represented at the circuit level. 
The generic gate $d_c^\mu\!=\!\exp(i\lambda\hat P_c^\mu A_c^\mu\hat P_c^\mu)$ is a three-qubit unitary whose nontrivial action is confined to the forbidden-pattern subspace selected by $\hat P_c^\mu$. 
For a first hardware-oriented implementation, one can restrict to diagonal projected phases. 
Each forbidden pattern $\tau$ is selected by conjugating with $X$ gates whenever $\tau_i=0$, applying a multi-controlled phase to the resulting all-one pattern, and uncomputing the $X$ gates. 
Because allowed patterns are not addressed by these controlled phases, the decoration remains exactly invisible on the protected Rule-54 trajectory.

\section{Conclusion and outlook}
\label{sec:conclusion}

We have constructed an exactly solvable family of QMBS in a decorated Rule-54 QCA. The central mechanism is to use the native soliton structure of Rule 54 as the protected skeleton and to decorate the remaining Hilbert space by local projected unitaries. In the bare Rule-54 circuit, the hard-core even-dimer sector is mapped to the odd-dimer sector and back with a coarse-grained translation. By inserting decorations that vanish on the sector reached after each half-step, the decorated Floquet unitary preserves this translation dynamics exactly, while remaining nontrivial outside the protected path.
The resulting scar sector is analytically controlled. Its basis states are labelled by hard-core configurations on a coarse-grained cycle, and its dimension is the Fibonacci number $F_{L-1}+F_{L+1}$. Thus the number of exact scar eigenstates grows exponentially with system size, while their fraction in the full qubit Hilbert space vanishes exponentially. This places the construction in the regime of weak ergodicity breaking rather than Hilbert-space fragmentation. The protected dynamics reduces to a finite permutation, and the exact scar eigenstates are obtained by Fourier transforming over the corresponding translation orbits. These orbit-Fourier states have explicitly known quasienergies and obey a sub-volume-law entanglement bound $S_A\!\leq\!\log p_\nu\leq\log(N/2)$ for arbitrary bipartitions.
For generic site-dependent non-diagonal projected decorations, finite-size simulations show a high-entanglement bulk, rapid entanglement growth and fidelity decay from outside product states, and quasienergy spacing statistics compatible with the random-matrix class after removing the protected sector. These observations support the picture of exact scars embedded in a thermal-like Floquet background.
This work also gives a qubit-QCA route to exact many-body scars that is different from SWAP-based or enlarged-alphabet constructions. The protected subspace is not obtained by restricting to a proper onsite alphabet, which is unavailable for qubits. Instead, the hard-core Rule-54 constraint creates an exponentially large constrained sector from qubit degrees of freedom alone. The construction therefore shows how deterministic finite-orbit structures of reversible automata can be converted into exact scar sectors by local projector decorations.

Several directions are natural. First, the same strategy can be applied to other reversible quantum cellular automaton with finite-orbit, solitonic, or kinetically constrained sectors. Rule 54 is a minimal interacting example, but the logic only requires a locally identifiable protected trajectory and a set of local forbidden-pattern projectors. This suggests a broader classification problem: which QCA admit exponentially large constrained sectors whose internal dynamics is a solvable permutation, and which local decorations make the complement thermal-like?
Second, the decorated Rule-54 circuit provides a useful setting for studying operator spreading and quantum nonstabilizerness~\cite{BravyiKitaev2005,Veitch2014,HowardCampbell2017,SeddonCampbell2019,LeoneOlivieroHamma2022,HaugPiroli2023,Iannotti2025,ViscardiDalmonteHammaTirrito2025,Collura2026,Korbany2025Magic,QianWang2025,HuangLiLeeZhong2025}. The protected orbit-Fourier scars have low entanglement because their support contains only finitely many product configurations, but the decorated complement can generate complex many-body dynamics. It would be interesting to quantify how stabilizer R\'enyi entropy, operator stabilizer entropy, and other magic diagnostics spread from different classes of initial states: protected orbit states, coherent orbit-Fourier scars, partially overlapping states, and generic outside product states. Such diagnostics could reveal whether the scar sector suppresses not only entanglement growth but also the growth of nonstabilizer resources, while the complement develops extensive magic under the same local circuit.
Third, monitored and open-system variants of the construction are promising~\cite{LiChenFisher2018,SkinnerRuhmanNahum2019,ChanNandkishorePretkoSmith2019,BaoChoiAltman2020,IppolitiKhemani2021,HuangLiZhaoLiuZhong2024,LiZhongYu2025,ZhaoHuangZhangLiZhong2025}. Since the protected sector is defined by local forbidden-pattern projectors, one can imagine measuring precisely these local constraints during the evolution. Weak monitoring, postselection, or feedback based on the forbidden-pattern outcomes may either stabilize the scar sector, induce leakage out of it, or generate measurement-induced transitions in the entanglement dynamics. A monitored decorated Rule-54 QCA would provide a controlled platform where exact soliton orbits, projector-selected sectors, and measurement-induced entanglement growth can be studied within the same circuit architecture.
Related deformations could also connect this projector-decorated QCA viewpoint to non-Hermitian localization, skin effects, and digital simulations of non-Hermitian topology~\cite{HatanoNelson1996,Lee2016NH,YaoWang2018,OkumaKawabataShiozakiSato2020,AshidaGongUeda2020,XiaoLee2023QuantumTopology,LiWanZhong2024,WangLiZhong2025}.
Fourth, the digital implementation deserves further exploration. The local Rule-54 gate can be decomposed into elementary reversible gates, and hardware-friendly diagonal decorations can be implemented as pattern-controlled phases on forbidden local configurations. The most accessible near-term observables are bitstring-level signatures: one-soliton translation, return probability, protected-orbit probability, and dimer-density profiles. These do not require full tomography and can be tested first in noiseless and noisy circuit simulators before moving to small-scale quantum processors. A full hardware demonstration would require assessing gate depth, compilation overhead, readout errors, and the robustness of the protected dynamics under imperfect decorations.
Finally, an important theoretical question is the stability of the construction. The exact scars rely on exact projector invisibility, so generic perturbations that violate the local constraints will eventually induce leakage. However, structured perturbations that preserve the forbidden-pattern kernel, diagonal phase deformations, or monitored correction schemes may retain a long-lived scarred regime. Understanding which imperfections destroy the scar sector immediately and which only broaden it into a prethermal protected manifold would connect the exact construction developed here to experimentally realistic scarred dynamics.

In summary, the decorated Rule-54 QCA shows that exact many-body scars can be engineered from native soliton orbits of an interacting qubit automaton. The model combines exact analytic control of an exponentially large scar sector with a thermal-like decorated complement, and it offers a concrete starting point for studying scarred dynamics, magic growth, monitored entanglement, and digital quantum simulation in reversible cellular automata.

\begin{acknowledgements}
\noindent
H.-Z. Li would like to thank Ching Hua Lee for the friendly guidance, and acknowledges the helpful initial discussions with Shuo Liu and Xue-Jia Yu. J.-X.\ Zhong was supported by the National Natural Science Foundation of China (Grant Nos.\ 12374046 and 11874316), the Shanghai Science and Technology Innovation Action Plan (Grant No.\ 24LZ1400800), the National Basic Research Program of China (Grant No.\ 2015CB921103), and the Program for Changjiang Scholars and Innovative Research Teams in Universities (Grant No.\ IRT13093). H.-Z.\ Li is supported by the China Scholarship Council (CSC) Scholarship (Grant No.\ 202506890103).
\end{acknowledgements}

%%%%%%%%%%%%%%%%% APPENDIX %%%%%%%%%%%%%%%%%%%

\clearpage
\newpage
\onecolumngrid

\appendix

\section{Rule-54 quantum cellular automaton}
\label{app:rule54_translation}

This appendix gives a detailed introduction to the Rule-54 QCA used in the main text and proves the hard-core soliton translation identity.  The model is a reversible, discrete-time quantum circuit on a one-dimensional qubit chain.  We label the physical sites by $j\!=\!0,1,\ldots,N-1$ and impose periodic boundary conditions.  The computational basis at site $j$ is denoted by $|s_j\rangle_j$, with $s_j\!\in\!\{0,1\}$.  Throughout the paper we take $N\!=\!2L$ even.

The elementary Rule-54 gate is a three-site permutation unitary.  On a local triplet $(j-1,j,j+1)$ it acts as
\begin{align}
R_j|x,y,z\rangle_{j-1,j,j+1}
=
|x,\chi(x,y,z),z\rangle_{j-1,j,j+1},
\label{eq:app_rule54_local_gate}
\end{align}
where
\begin{align}
\chi(x,y,z)=(x+y+z+xz)\bmod 2.
\label{eq:app_rule54_boolean}
\end{align}
Thus the neighboring qubits $x$ and $z$ are copied through unchanged, while the central qubit $y$ is updated by a reversible Boolean rule.  Equivalently, the local matrix elements are
\begin{align}
\langle x',y',z'|R_j|x,y,z\rangle
=
\delta_{x',x}\delta_{z',z}\delta_{y',\chi(x,y,z)}.
\label{eq:app_rule54_matrix_elements}
\end{align}
The truth table of the central update is
\begin{align}
\begin{array}{ccccccccc}
x & 0 & 0 & 0 & 0 & 1 & 1 & 1 & 1 \\
y & 0 & 0 & 1 & 1 & 0 & 0 & 1 & 1 \\
z & 0 & 1 & 0 & 1 & 0 & 1 & 0 & 1 \\
\hline
\chi(x,y,z) & 0 & 1 & 1 & 0 & 1 & 1 & 0 & 0
\end{array}
\label{eq:app_rule54_truth_table}
\end{align}
This is the quantum version of the classical elementary cellular automaton Rule 54, represented as a unitary permutation of computational-basis states. Also it can be depicted by the map:
\begin{align}
    \vcenter{\hbox{\includegraphics[width=0.5\textwidth]{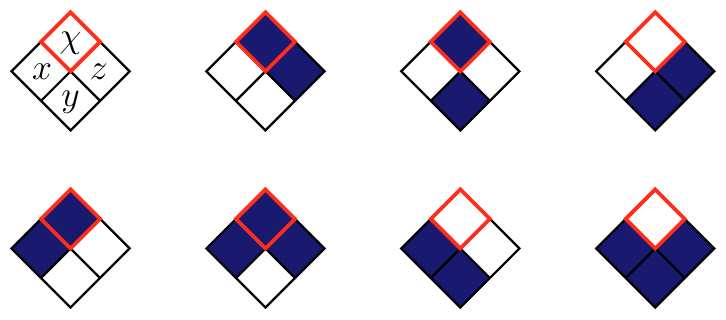}}}
\end{align}

The local gate is its own inverse.  Indeed,
\begin{align}
\chi(x,\chi(x,y,z),z)
&=
\left[x+(x+y+z+xz)+z+xz\right]\bmod 2
\nonumber\\
&=
\left[y+2x+2z+2xz\right]\bmod 2
\nonumber\\
&=y.
\label{eq:app_rule54_involution}
\end{align}
Therefore $R_j^2\!=\!\mathbb I$.  Since $R_j$ permutes computational-basis states bijectively, it is unitary. A full Rule-54 Floquet step is constructed from staggered even- and odd-center layers,
\begin{align}
U_e=\prod_{r=0}^{L-1}R_{2r},
\quad
U_o=\prod_{r=0}^{L-1}R_{2r+1},
\quad
U_{54}=U_oU_e.
\label{eq:app_rule54_floquet_layers}
\end{align}
Within each layer the gates commute.  This is because gates with the same center parity update distinct central sites; possible shared neighboring sites are used only as controls and are copied through.  The convention $U_{54}=U_oU_e$ fixes a direction for the coarse-grained translation.  Reversing the convention reverses this direction but does not change the construction.

We now prove the protected hard-core translation identity.  Let the coarse-grained lattice have length $L$, and let $a_r\!\in\!\{0,1\}$ denote the occupation of the even bond $(2r,2r+1)$.  The hard-core configuration space is
\begin{align}
\mathcal I_L=
\{\mathbf a=(a_0,\ldots,a_{L-1})\in\{0,1\}^L:
 a_ra_{r+1}=0\ \text{for all }r\},
\label{eq:app_IL_def}
\end{align}
with indices modulo $L$.  The even- and odd-dimer product states are
\begin{align}
|\mathbf a\rangle_e
&=
\bigotimes_{r=0}^{L-1}|a_r\rangle_{2r}|a_r\rangle_{2r+1},
\\
|\mathbf a\rangle_o
&=
\bigotimes_{r=0}^{L-1}|a_r\rangle_{2r+1}|a_r\rangle_{2r+2}.
\label{eq:app_even_odd_dimer_states}
\end{align}
An occupied coarse-grained site corresponds to a physical $11$ dimer, while an empty coarse-grained site corresponds to a physical $00$ dimer.  The hard-core condition forbids neighboring occupied dimers on the coarse-grained cycle.
Let $T$ be the coarse-grained translation defined by $(T\mathbf a)_r\!=\!a_{r-1}$.  We prove
\begin{align}
U_e|\mathbf a\rangle_e=|\mathbf a\rangle_o,
\quad
U_o|\mathbf a\rangle_o=|T\mathbf a\rangle_e,
\quad
U_{54}|\mathbf a\rangle_e=|T\mathbf a\rangle_e.
\label{eq:app_translation_identity_full}
\end{align}
For $|\mathbf a\rangle_e$, the physical qubits are $s_{2r}=a_r$ and $s_{2r+1}=a_r$.  Around the even center $2r$, the triplet is $(a_{r-1},a_r,a_r)$.  Applying the even layer gives
\begin{align}
s'_{2r}
&=\chi(a_{r-1},a_r,a_r)
\nonumber\\
&=(a_{r-1}+a_r+a_r+a_{r-1}a_r)\bmod2
\nonumber\\
&=(a_{r-1}+a_{r-1}a_r)\bmod2
\nonumber\\
&=a_{r-1},
\label{eq:app_even_layer_update}
\end{align}
where the last equality uses $a_{r-1}a_r\!=\!0$.  Since odd sites are not updated in the even layer, $s'_{2r+1}\!=\!a_r$.  Hence the output qubits on the physical bond $(2r,2r+1)$ are
\begin{align}
s'_{2r}=a_{r-1},
\quad
s'_{2r+1}=a_r,
\label{eq:app_even_layer_output}
\end{align}
which is precisely the odd-dimer state $|\mathbf a\rangle_o$.

Now start from $|\mathbf a\rangle_o$.  Then $s_{2r}\!=\!a_{r-1}$ and $s_{2r+1}\!=\!a_r$.  Around the odd center $2r+1$, the triplet is again $(a_{r-1},a_r,a_r)$.  Thus
\begin{align}
s''_{2r+1}=\chi(a_{r-1},a_r,a_r)=a_{r-1}.
\label{eq:app_odd_layer_update}
\end{align}
The even site $2r$ is not updated in $U_o$, so $s''_{2r}\!=\!a_{r-1}$.  Therefore
\begin{align}
s''_{2r}=s''_{2r+1}=a_{r-1}=(T\mathbf a)_r,
\label{eq:app_odd_layer_output}
\end{align}
which is exactly the even-dimer state $|T\mathbf a\rangle_e$.  This proves Eq.~\eqref{eq:app_translation_identity_full}.
The even hard-core soliton subspace is
\begin{align}
\mathcal T_e={\rm span}\{|\mathbf a\rangle_e:\mathbf a\in\mathcal I_L\}.
\label{eq:app_Te_def}
\end{align}
Eq.~\eqref{eq:app_translation_identity_full} implies
\begin{align}
U_{54}|_{\mathcal T_e}=T|_{\mathcal T_e}.
\label{eq:app_U54_equals_T}
\end{align}
Thus the hard-core sector is invariant under the bare Rule-54 Floquet step, and the restricted dynamics is exactly a translation on the coarse-grained lattice.

\section{Allowed local patterns}
\label{app:projector_images}
This appendix proves that the local pattern projectors used in the main text identify precisely the even- and odd-dimer hard-core sectors.

For a center site $c$, a local three-qubit pattern is denoted by $\tau\!=\!(\tau_L,\tau_C,\tau_R)$ on the sites $(c-1,c,c+1)$.  For the even-dimer hard-core sector, the allowed local patterns are
\begin{align}
X_{2r}^{E}=\{000,100,011\},
\quad
X_{2r+1}^{E}=\{000,001,110\}.
\label{eq:app_XE}
\end{align}
For the odd-dimer hard-core sector, the pattern sets are shifted by one site,
\begin{align}
X_{2r}^{O}=\{000,001,110\},
\quad
X_{2r+1}^{O}=\{000,100,011\}.
\label{eq:app_XO}
\end{align}
The local forbidden-pattern projector is
\begin{align}
P_c^\mu
=
\mathbb I_{c-1,c,c+1}
-
\sum_{\tau\in X_c^\mu}|\tau\rangle\langle\tau|_{c-1,c,c+1},
\quad \mu=E,O,
\label{eq:app_local_projector}
\end{align}
and $\hat P_c^\mu$ denotes its embedding into the full chain.  We define
\begin{align}
\Pi_e=\prod_{c=0}^{N-1}(\mathbb I-\hat P_c^E),
\quad
\Pi_o=\prod_{c=0}^{N-1}(\mathbb I-\hat P_c^O).
\label{eq:app_global_projectors}
\end{align}
Since all $P_c^\mu$ are diagonal in the computational basis, the products are unambiguous.

Then, we first show that $\Pi_e$ projects onto $\mathcal T_e$.  If $|\mathbf a\rangle_e\!\in\!\mathcal T_e$, then at an even center $2r$ the local triplet is
\begin{align}
(s_{2r-1},s_{2r},s_{2r+1})=(a_{r-1},a_r,a_r).
\label{eq:app_even_center_triplet}
\end{align}
Because $a_{r-1}a_r\!=\!0$, this triplet can only be $000$, $100$, or $011$.  Hence it lies in $X_{2r}^{E}$.  At an odd center $2r+1$, the triplet is
\begin{align}
(s_{2r},s_{2r+1},s_{2r+2})=(a_r,a_r,a_{r+1}),
\label{eq:app_odd_center_triplet}
\end{align}
which, using $a_ra_{r+1}\!=\!0$, can only be $000$, $001$, or $110$.  Hence it lies in $X_{2r+1}^{E}$.  Therefore $\hat P_c^E|\mathbf a\rangle_e\!=\!0$ for every $c$, and $\Pi_e|\mathbf a\rangle_e\!=\!|\mathbf a\rangle_e$.

Conversely, suppose a computational-basis state $|s_0\ldots s_{N-1}\rangle$ is not annihilated by any allowed-pattern projector, i.e. all its local triplets lie in the sets $X_c^E$.  From the condition at the odd center $2r+1$, the allowed triplets $000$, $001$, and $110$ all have the first two entries equal.  Therefore
\begin{align}
s_{2r}=s_{2r+1}
\label{eq:app_even_dimer_forced}
\end{align}
for all $r$.  Thus the state has the even-dimer form for some coarse-grained configuration $\mathbf a$ with $a_r\!=\!s_{2r}\!=\!s_{2r+1}$.  The condition at the even center $2r$ then says that the triplet $(a_{r-1},a_r,a_r)$ must lie in $\{000,100,011\}$.  If $a_r\!=\!1$, the only allowed triplet with the last two entries equal to $1$ is $011$, so $a_{r-1}\!=\!0$.  Hence $a_{r-1}a_r\!=\!0$ for all $r$.  Equivalently, $\mathbf a\!\in\!\mathcal I_L$.  Therefore the image of $\Pi_e$ is exactly $\mathcal T_e$:
\begin{align}
{\rm Im}\,\Pi_e=\mathcal T_e.
\label{eq:app_image_Pi_e}
\end{align}

The proof for $\Pi_o$ is identical after shifting all site labels by one.  Thus
\begin{align}
{\rm Im}\,\Pi_o=\mathcal T_o.
\label{eq:app_image_Pi_o}
\end{align}
Eqs.~\eqref{eq:app_image_Pi_e} and \eqref{eq:app_image_Pi_o} justify the projector identities used in the main text.

\section{Proof of protected decorated dynamics and locality}
\label{app:decorated_proof}

This appendix proves the exact protected action of the decorated circuit.  The decorated local generators are
\begin{align}
G_c^{(e)}=\hat P_c^O A_c^{(e)}\hat P_c^O,
\quad
G_c^{(o)}=\hat P_c^E A_c^{(o)}\hat P_c^E,
\label{eq:app_decorated_generators}
\end{align}
where $A_c^{(e)}$ and $A_c^{(o)}$ are arbitrary three-qubit Hermitian operators.  The corresponding gates are
\begin{align}
d_c^{(e)}=\exp(i\lambda_eG_c^{(e)}),
\quad
d_c^{(o)}=\exp(i\lambda_oG_c^{(o)}),
\label{eq:app_decorated_local_gates}
\end{align}
and the decoration layers are
\begin{align}
D_e=\overleftarrow{\prod_{c=0}^{N-1}}d_c^{(e)},
\quad
D_o=\overleftarrow{\prod_{c=0}^{N-1}}d_c^{(o)}.
\label{eq:app_decorated_layers}
\end{align}
The decorated Floquet unitary is
\begin{align}
U_{\rm DR54}=D_oU_oD_eU_e.
\label{eq:app_udr54}
\end{align}

Then, let $|\psi_o\rangle\!\in\!\mathcal T_o$.  Since ${\rm Im}\,\Pi_o\!=\!\mathcal T_o$, every local forbidden projector satisfies
\begin{align}
\hat P_c^O|\psi_o\rangle=0.
\label{eq:app_PO_annihilates_To}
\end{align}
Thus
\begin{align}
G_c^{(e)}|\psi_o\rangle
=
\hat P_c^O A_c^{(e)}\hat P_c^O|\psi_o\rangle
=0.
\label{eq:app_Ge_annihilates_To}
\end{align}
Since all positive powers of $G_c^{(e)}$ annihilate $|\psi_o\rangle$, the exponential acts as the identity:
\begin{align}
d_c^{(e)}|\psi_o\rangle
=
\exp(i\lambda_eG_c^{(e)})|\psi_o\rangle
=|\psi_o\rangle.
\label{eq:app_de_identity_on_To_local}
\end{align}
Therefore
\begin{align}
D_e|_{\mathcal T_o}=\mathbb I_{\mathcal T_o}.
\label{eq:app_De_identity_on_To}
\end{align}
The same argument gives
\begin{align}
D_o|_{\mathcal T_e}=\mathbb I_{\mathcal T_e}.
\label{eq:app_Do_identity_on_Te}
\end{align}
Combining these identities with Eq.~\eqref{eq:app_translation_identity_full}, for every $\mathbf a\!\in\!\mathcal I_L$ we obtain
\begin{align}
U_{\rm DR54}|\mathbf a\rangle_e
&=D_oU_oD_eU_e|\mathbf a\rangle_e
\nonumber\\
&=D_oU_oD_e|\mathbf a\rangle_o
\nonumber\\
&=D_oU_o|\mathbf a\rangle_o
\nonumber\\
&=D_o|T\mathbf a\rangle_e
\nonumber\\
&=|T\mathbf a\rangle_e.
\label{eq:app_decorated_translation_proof}
\end{align}
Hence
\begin{align}
U_{\rm DR54}|_{\mathcal T_e}=T|_{\mathcal T_e}.
\label{eq:app_decorated_translation_identity}
\end{align}
Since the image of $\Pi_e$ is $\mathcal T_e$, Eq.~\eqref{eq:app_decorated_translation_identity} also implies the no-leakage identity
\begin{align}
(\mathbb I-\Pi_e)U_{\rm DR54}\Pi_e=0.
\label{eq:app_no_leakage_identity}
\end{align}
Because $U_{\rm DR54}$ is unitary and $\mathcal T_e$ is invariant, the orthogonal complement $\mathcal T_e^\perp$ is also invariant.

Finally, we comment on circuit depth.  Each decoration gate has support on three consecutive physical sites.  Gates with overlapping supports need not commute, but the ordering in Eq.~\eqref{eq:app_decorated_layers} is part of the model definition and may be chosen for implementation convenience.  If one chooses a finite-color ordering of the center sites, the three-site supports can be arranged into a constant number of parallel sublayers.  On a one-dimensional ring, the overlap graph of length-three intervals has bounded degree, and a fixed number of colors independent of $N$ is sufficient.  Thus the decorated circuit remains a local finite-depth Floquet circuit.

\section{Counting the Fibonacci scar sector}
\label{app:fibonacci_counting}

We give the detailed counting of the hard-core configuration space $\mathcal I_L$.  Recall that $\mathcal I_L$ is the set of binary configurations $\mathbf a\!=\!(a_0,\ldots,a_{L-1})$ on a cycle, subject to $a_ra_{r+1}\!=\!0$ with indices modulo $L$.

\subsection{Transfer-matrix count}
The compatibility matrix for two neighboring occupations is
\begin{align}
M_{ab}
=
\begin{cases}
1, & ab=0,\\
0, & ab=1,
\end{cases}
\quad
M=
\begin{pmatrix}
1&1\\
1&0
\end{pmatrix}.
\label{eq:app_M_def}
\end{align}
For a periodic chain, a configuration contributes one unit to the count precisely when every adjacent pair is compatible, including the pair $(a_{L-1},a_0)$.  Therefore
\begin{align}
|\mathcal I_L|
&=
\sum_{a_0,\ldots,a_{L-1}=0}^{1}
M_{a_0a_1}M_{a_1a_2}\cdots M_{a_{L-2}a_{L-1}}M_{a_{L-1}a_0}
\nonumber\\
&=
\sum_{a_0=0}^{1}(M^L)_{a_0a_0}
\nonumber\\
&=
{\rm Tr}\,M^L.
\label{eq:app_trace_derivation}
\end{align}
The characteristic polynomial of $M$ is $\lambda^2-\lambda-1\!=\!0$, so the two eigenvalues are $\lambda_+\!=\!\varphi$ and $\lambda_-\!=\!-\varphi^{-1}$.  Hence
\begin{align}
|\mathcal I_L|
&={\rm Tr}\,M^L
\nonumber\\
&=
\lambda_+^L+\lambda_-^L
\nonumber\\
&=
\varphi^L+(-\varphi^{-1})^L.
\label{eq:app_lucas_derivation}
\end{align}
Using Binet's formula $F_n\!=\!(\varphi^n-(-\varphi^{-1})^n)/\sqrt5$, one obtains
\begin{align}
F_{L-1}+F_{L+1}
&=
\frac{\varphi^{L-1}+\varphi^{L+1}-(-\varphi^{-1})^{L-1}-(-\varphi^{-1})^{L+1}}{\sqrt5}
\nonumber\\
&=\varphi^L+(-\varphi^{-1})^L.
\label{eq:app_fibonacci_lucas_identity}
\end{align}
Thus
\begin{align}
|\mathcal I_L|=F_{L-1}+F_{L+1}.
\label{eq:app_lucas_fibonacci_final}
\end{align}
This is the Fibonacci count of the cycle, written in Fibonacci form.

\subsection{Fixed-soliton-number count}

Let $\mathcal I_{L,k}$ be the subset of $\mathcal I_L$ with exactly $k$ occupied coarse-grained sites.  For $k=0$, $|\mathcal I_{L,0}|\!=\!1$.  For $1\!\leq\!k\!\leq\!\lfloor L/2\rfloor$, the number of ways to place $k$ indistinguishable hard-core particles on a labelled cycle of length $L$ is
\begin{align}
|\mathcal I_{L,k}|
=
\frac{L}{L-k}\binom{L-k}{k}.
\label{eq:app_fixed_k_count}
\end{align}
To derive this, mark one occupied site as a root.  Moving clockwise from the root, let $g_j\!\geq\!1$ be the number of empty sites after the $j$-th occupied site and before the next occupied site.  Then
\begin{align}
g_1+g_2+\cdots+g_k=L-k,
\quad
 g_j\geq1.
\label{eq:app_gap_constraints}
\end{align}
The number of positive integer solutions is $\binom{L-k-1}{k-1}$.  There are $L$ possible positions for the rooted occupied site, but each unrooted configuration with $k$ occupied sites is counted $k$ times.  Therefore
\begin{align}
|\mathcal I_{L,k}|
&=
\frac{L}{k}\binom{L-k-1}{k-1}
\nonumber\\
&=
\frac{L}{L-k}\binom{L-k}{k}.
\label{eq:app_fixed_k_count_derivation}
\end{align}
Summing Eq.~\eqref{eq:app_fixed_k_count} over $k$ recovers Eq.~\eqref{eq:app_lucas_fibonacci_final}.

\section{Orbit-Fourier eigenstates and degeneracies}
\label{app:orbit_fourier}

Here we give the detailed proof that the orbit-Fourier states are exact eigenstates of the decorated Floquet unitary and explain how degeneracies should be handled in exact diagonalization.

Let
\begin{align}
\mathcal O_\nu=\{\mathbf a_\nu,T\mathbf a_\nu,\ldots,T^{p_\nu-1}\mathbf a_\nu\}
\label{eq:app_orbit_def}
\end{align}
be an orbit of minimal period $p_\nu$.  Since the decorated circuit acts as translation on $\mathcal T_e$, one has
\begin{align}
U_{\rm DR54}|T^t\mathbf a_\nu\rangle_e
=
|T^{t+1}\mathbf a_\nu\rangle_e,
\quad
 t=0,\ldots,p_\nu-1,
\label{eq:app_orbit_shift}
\end{align}
where $t+1$ is understood modulo $p_\nu$.  Define
\begin{align}
|\Phi_{\nu,m}\rangle
=
\frac{1}{\sqrt{p_\nu}}
\sum_{t=0}^{p_\nu-1}
\omega_{p_\nu}^{-mt}|T^t\mathbf a_\nu\rangle_e,
\quad
\omega_{p_\nu}=e^{2\pi i/p_\nu}.
\label{eq:app_fourier_state}
\end{align}
Then
\begin{align}
U_{\rm DR54}|\Phi_{\nu,m}\rangle
&=
\frac{1}{\sqrt{p_\nu}}
\sum_{t=0}^{p_\nu-1}
\omega_{p_\nu}^{-mt}|T^{t+1}\mathbf a_\nu\rangle_e
\nonumber\\
&=
\frac{1}{\sqrt{p_\nu}}
\sum_{t'=0}^{p_\nu-1}
\omega_{p_\nu}^{-m(t'-1)}|T^{t'}\mathbf a_\nu\rangle_e
\nonumber\\
&=
\omega_{p_\nu}^{m}|\Phi_{\nu,m}\rangle.
\label{eq:app_fourier_eigenproof}
\end{align}
Thus the eigenphase is $2\pi m/p_\nu$ modulo $2\pi$.
The orbit-Fourier states are orthonormal.  If $\nu\!\neq\!\nu'$, then the supports of the two states in the computational basis are disjoint.  If $\nu\!=\!\nu'$, then
\begin{align}
\langle\Phi_{\nu,m}|\Phi_{\nu,m'}\rangle
&=
\frac{1}{p_\nu}
\sum_{t,t'=0}^{p_\nu-1}
\omega_{p_\nu}^{mt-m't'}
{}_e\langle T^t\mathbf a_\nu|T^{t'}\mathbf a_\nu\rangle_e
\nonumber\\
&=
\frac{1}{p_\nu}
\sum_{t=0}^{p_\nu-1}
\omega_{p_\nu}^{(m-m')t}
=\delta_{m,m'}.
\label{eq:app_fourier_orthogonality}
\end{align}
Since the total number of Fourier states is $\sum_\nu p_\nu\!=\!|\mathcal I_L|\!=\!\dim\mathcal T_e$, they form an orthonormal basis of $\mathcal T_e$.

Degeneracies occur whenever distinct pairs $(\nu,m)$ and $(\nu',m')$ yield the same phase.  This is common because different orbits can have the same period, and orbits of different periods can also share phases when $m/p_\nu\!=\!m'/p_{\nu'}$ modulo one.  A dense eigensolver is free to output arbitrary orthonormal combinations within such degenerate eigenspaces.  Therefore individual numerical eigenvectors need not coincide with the analytic orbit-Fourier states.  The invariant diagnostic is the protected-sector weight
\begin{align}
w_\alpha
=
\langle\phi_\alpha|\Pi_e|\phi_\alpha\rangle,
\label{eq:app_protected_weight}
\end{align}
where $|\phi_\alpha\rangle$ is a numerical eigenvector of $U_{\rm DR54}$.  For exact protected eigenstates, or for any linear combination of degenerate protected eigenstates, $w_\alpha\!=\!1$.  For bulk eigenstates outside the protected sector, $w_\alpha\!=\!0$ in exact arithmetic, up to numerical roundoff and accidental degeneracy mixing.

\section{Entanglement bound for orbit-Fourier scars}
\label{app:entanglement_bound}

We prove the entanglement bound used in the main text.  Let $A\cup\bar A$ be any bipartition of the physical chain.  Since every $|T^t\mathbf a_\nu\rangle_e$ is a computational-basis product state, it factorizes as
\begin{align}
|T^t\mathbf a_\nu\rangle_e
=
|\ell_t\rangle_A\otimes|r_t\rangle_{\bar A}.
\label{eq:app_product_factorization}
\end{align}
Substituting this into the orbit-Fourier state gives
\begin{align}
|\Phi_{\nu,m}\rangle
=
\frac{1}{\sqrt{p_\nu}}
\sum_{t=0}^{p_\nu-1}
\omega_{p_\nu}^{-mt}|\ell_t\rangle_A\otimes|r_t\rangle_{\bar A}.
\label{eq:app_schmidt_rank_expression}
\end{align}
This is not necessarily the Schmidt decomposition, because the restricted states $|\ell_t\rangle_A$ and $|r_t\rangle_{\bar A}$ can repeat or fail to be mutually orthogonal.  Nevertheless, it is an expansion containing at most $p_\nu$ product terms.  Hence the Schmidt rank obeys
\begin{align}
\chi_A(|\Phi_{\nu,m}\rangle)\leq p_\nu.
\label{eq:app_schmidt_rank_bound}
\end{align}
Let $\{\lambda_j\}_{j=1}^{\chi_A}$ be the nonzero eigenvalues of the reduced density matrix $\rho_A$.  The von Neumann entropy satisfies
\begin{align}
S_A
&=-\sum_{j=1}^{\chi_A}\lambda_j\log\lambda_j
\leq \log\chi_A
\leq \log p_\nu
\leq \log L.
\label{eq:app_von_neumann_bound}
\end{align}
The first inequality is the standard fact that the entropy of a probability distribution supported on $\chi_A$ points is maximized by the uniform distribution.

For R\'enyi entropies $S_A^{(\alpha)}\!=\!(1-\alpha)^{-1}\log {\rm Tr}\,\rho_A^\alpha$, with $\alpha\!>\!0$, the same rank bound gives
\begin{align}
S_A^{(\alpha)}
\leq \log\chi_A
\leq \log p_\nu
\leq \log L.
\label{eq:app_renyi_bound}
\end{align}
Thus the bound holds for arbitrary bipartitions and for all positive R\'enyi indices.  Since $L=N/2$, this is the logarithmic sub-volume-law bound $S_A^{(\alpha)}\!\leq\!\log(N/2)$ quoted in the main text.

% In numerical diagonalization, degenerate scar eigenphases may lead to arbitrary linear combinations of orbit-Fourier states.  Such combinations can have larger entanglement than a particular orbit-Fourier representative if several degenerate protected states are mixed.  This does not contradict the analytic bound: the protected degenerate eigenspace always admits an orbit-Fourier basis satisfying the bound.  In figures, one should therefore overlay analytically constructed orbit-Fourier scars or identify the full protected subspace by the projector $\Pi_e$, rather than relying only on individual eigenvectors returned by a generic eigensolver.

\section{Numerical implementation details}
\label{app:numerics}

All numerics use computational-basis permutation matrices for the Rule-54 layers and embedded three-qubit projected random gates. Fig.~\ref{fig:2} uses dense exact diagonalization at $N\!=\!10$ for vacuum and one-soliton target-orbit constructions, with $\lambda\!=\!0.75$. Fig.~\ref{fig:3} uses the full hard-core projector construction at $N\!=\!8$, and $\lambda\!=\!0.85$, after removing the protected sector.

\end{document}